\documentclass[a4,11pt]{article}
\usepackage{fullpage}

\usepackage[dvips]{graphicx}
\graphicspath{{Figures/}} 
\usepackage{xcolor}
\usepackage[nice]{nicefrac}
\usepackage{amsmath}
\usepackage{mathtools}

\usepackage{amsthm}
\usepackage[colorlinks,linkcolor=blue,anchorcolor=blue,citecolor=blue,urlcolor=blue]{hyperref} 

\newtheorem{theorem}{Theorem}
\newtheorem{proposition}{Proposition}
\newtheorem{corollary}{Corollary}
\newtheorem{fact}{Fact}

\newcommand{\fa}{\mbox{\scriptsize \sf fa}}
\renewcommand{\root}{\mbox{root}}

\usepackage{todonotes}

\begin{document}

\title{Source-Oblivious Broadcast\thanks{A preliminary version of this work appeared in the proceedings of the 18th Annual Conference on Theory and Applications of Models of Computation (TAMC), Hong Kong, China, May 13-15, 2024.}}

\date{}

\author{Pierre Fraigniaud\thanks{Additional support from ANR project DUCAT (ref. ANR-20-CE48-0006).}\\~\\
{\small Inst. de Recherche en Informatique Fondamentale}\\
{\small CNRS and Université Paris Cité}\\
{\small Paris, France}
\and
Hovhannes A. Harutyunyan\\~\\
{\small Dep. of Comp. Science and Soft. Engineering}\\
{\small Concordia University}\\
{\small Montréal, Canada}
}

\maketitle            

\begin{abstract}
This paper revisits the study of (minimum) broadcast graphs, i.e., graphs enabling fast information dissemination from every source node to all the other nodes (and having minimum number of edges for this property). This study is performed in  the framework of compact distributed data structures, that is, when the broadcast protocols are bounded to be encoded at each node as an ordered list of neighbors specifying, upon reception of a message, in which order this message must be passed to these neighbors. We show that this constraint does not limit the power of broadcast protocols, as far as the design of (minimum) broadcast graphs is concerned. 
Specifically, we show that, for every~$n$, there are $n$-node graphs for which it is possible to design protocols encoded by lists yet enabling broadcast in $\lceil\log_2n\rceil$ rounds from every source, which is optimal even for general (i.e., non space-constrained) broadcast protocols.   Moreover, we show that, for every~$n$, there exist such graphs with the additional property that they are asymptotically as sparse as the sparsest graphs for which $\lceil\log_2n\rceil$-round broadcast protocols exist, up to a constant multiplicative factor. Concretely, these graphs have $O(n\cdot L(n))$ edges, where $L(n)$ is the number of leading~1s in the binary representation of $n-1$, and general minimum broadcast graphs are known to have $\Omega(n\cdot L(n))$ edges.

\medbreak
\end{abstract}

\newpage

\section{Introduction} 

\subsection{Context}

The \emph{broadcast} problem is the information dissemination problem consisting of passing a piece of information (i.e., an atomic message) from a node of a connected graph to all the nodes of the graph. A broadcast protocol proceeds in synchronous \emph{rounds}. Initially, a single node is informed, called the \emph{source} node.  At each round, every informed node (i.e., every node possessing the information) can transmit the information to at most one of its neighbors. It follows that the number of informed nodes can at most double at each round, and thus, for every source node~$s$, the minimum number of rounds required to broadcast from~$s$ in an $n$-node graph~$G$, denoted by $b(G,s)$, is at least $\lceil\log_2n\rceil$. 

\subsubsection{Broadcast Graphs}

For every $n$-node connected graph $G$, let 
\[
b(G)=\max_{s\in V(G)}b(G,s)
\] 
be the \emph{broadcast time} of~$G$. Since $b(G,s)\geq \lceil\log_2n\rceil$ for every source node~$s$, we have $b(G)\geq \lceil\log_2n\rceil$. For every $n\geq 1$, any  graphs~$G$ with $n$ nodes and satisfying  $b(G)=\lceil\log_2n\rceil$ is called \emph{broadcast graph}~\cite{FHMP79}. 
Observe that broadcast graphs do exist. In particular, for every $n\geq 1$, the complete graph $K_n$ is a broadcast graph. Indeed, in the complete graph, it is always possible to construct a matching from the set of informed nodes to the set of non-informed nodes, which saturates one of the two sets, and therefore $b(K_n)=\lceil\log_2n\rceil$. 

\subsubsection{Minimum Broadcast Graphs}

Motivated by the design of networks supporting efficient communication protocols but consuming few resources, the construction of sparse broadcast graphs has attracted lot of attention (cf. Section~\ref{subsec:relatedwork}). For every $n\geq 1$, let $B(n)$ denotes the minimum number of edges of broadcast graphs with $n$ nodes. As all complete graphs are broadcast graphs, $B(n)$ is well defined for all~$n$. However, there are broadcast graphs much sparser than  complete graphs. For instance, for every $d\geq 0$, the $d$-dimensional hypercube~$Q_d$, with $n=2^d$ nodes, is a broadcast graph as $b(Q_d)=d= \lceil\log_2n\rceil$ (an optimal broadcast protocol merely consists to transmit the information sequentially through edges of increasing dimension). As a consequence, for $n$ is a power of~2, we have $B(n)\leq \frac12 n\log_2n$.  In fact, it is known \cite{GrigniP91} that 
\begin{equation}\label{eq:B}
B(n)=\Theta(n \cdot L(n)),
\end{equation}
 where, for every positive integer~$x$, $L(x)$ denotes the number of consecutive leading~1s in the binary representation of $x-1$.  For instance, $L(12)=1$ as $11=(1011)_2$. In particular, if
$n=2^d$ for $d\geq 1$, then $B(n)=\Theta(n\log n)$ as $n-1=(11\dots11)_2$, and thus $L(n)=d=\log_2n$. In fact, it is easy to see (see~\cite{FHMP79}) that all hypercubes~$Q_d$, $d\geq 1$, are \emph{minimum} broadcast graphs, that is, broadcast graphs with the smallest number of edges --- this is simply because, for $n=2^d$, every source node must be active at each round~$r=1,\dots,d$ for insuring that all nodes receive the information after $d$ rounds, and hence every (source) node must have degree~$d$. Therefore, for $n$ power of~2, we have $B(n)=\frac12 n\log_2n$.

\subsection{Objective}

We are interested in the \emph{encoding} of broadcast protocols at each node of a graph. For any source node~$s$ of a graph~$G$, every node $v$ of $G$ receiving a piece of information originated from~$s$ must inform its neighbors (non necessarily all) in a right order for insuring that the information is broadcast fast, ideally in $\lceil\log_2n\rceil$ rounds whenever $G$ is a broadcast graph. If the local encoding of the protocol is done in a brute force manner, every node~$v$ stores a table $T_v$ with $n$ entries, one for each source~$s$, such that 
\[
T_v[s]=(u_{s,1},\dots,u_{s,k_s})
\]
provides $v$ with an ordered list of neighbors that $v$ must sequentially inform upon reception of a piece of information broadcast from~$s$. That is, upon receiving a message broadcast from~$s$, node~$v$ forwards that message to $u_{s,1}$ first, then to $u_{s,2}$, and so on, up to $u_{s,k_s}$, in $k_s$ successive rounds. In the worst case, this encoding may consume up to $O(n \log d!)$ bits to be stored at a degree-$d$ node~$v$, which can be almost as high as storing the entire graph~$G$ at~$v$ whenever $d=\Theta(n)$. 

\subsubsection{Source-Oblivious Broadcast}

With the objective of limiting the space complexity of encoding broadcast protocols locally at each node, we consider \emph{source-oblivious} broadcast protocols, as previously considered in, e.g.,~\cite{DiksP96,RosenthalS87,SlaterCH81}. Any such protocol can be encoded at each node~$v$ by a unique ordered list 
\begin{equation}\label{eq:thelists}
\ell_v=(u_1,\dots,u_k)
\end{equation}
of $k$ distinct neighbors of~$v$, where $k\leq \deg(v)$, hence consuming only $O(d\log d)$ bits at degree-$d$ nodes.  That is, upon receiving a message, node~$v$ forwards that message to $u_1$ first, then to $u_2$, and so on, up to $u_k$, in $k$ successive rounds, no matter the source of the information is.

\subsubsection{Fully-Adaptive Source-Oblivious Broadcast}

We actually focus on the variant of source-oblivious broadcast introduced in~\cite{GholamiH22a,GholamiH22b}, called \emph{fully-adaptive}. Specifically, we assume that, upon reception of a piece of information broadcast from a source node~$s$, every node~$v$ acknowledges reception by sending a signal message to all its neighbors. Note that the signal messages are short in comparison to the broadcast messages, which could be arbitrarily large. It follows that, at the end of each round, every node~$v$ is aware of which of its neighbors have received the information, and this holds even if the node~$v$ has not yet received that information. 

Broadcast thus performs as follows in the fully-adaptive source-oblivious model. Upon receiving a piece of information originated from any source~$s$, every node~$v$ initiates a series of \emph{calls} to its neighbors during subsequent rounds. Let $\ell_v=(u_1,\dots,u_k)$ be the list of node~$v$, and assume that $v$ received the broadcast message at round~$r$. Then, for $i=1,\dots,k$, node~$v$ aims at forwarding the message to node~$u_i$ at round $r+i$. However, if node~$u_i$ has already received the information at the round when $v$ is supposed to send the message to~$u_i$, then $u_i$ is skipped, and the message is transmitted to $u_{i+1}$ instead, unless $u_{i+1}$  is also already informed, in which case~$u_{i+2}$ is considered, etc. More generally, at a given round, node~$v$ forwards the broadcast message to the next node in its list~$\ell_v$ that has not already received that message. It stops when the list is exhausted. 

\subsection{Our Results}

In a nutshell, we show that constraining broadcast protocols by bounding them to be source-oblivious does not limit their power, as far as the design of broadcast graphs and minimum broadcast graphs is concerned. Specifically, we first establish the following. 

\begin{theorem}\label{theo:BG}
In the fully-adaptive source-oblivious model,  there are $n$-node broadcast graphs for every $n\geq 1$, i.e., $n$-node graphs for which there exists a collection of lists $(\ell_v)_{v\in V(G)}$ achieving broadcast  in $\lceil\log_2n\rceil$ rounds from any source node. In particular, for every ${n\geq 1}$, the broadcast time of the clique~$K_n$ is $\lceil\log_2n\rceil$ in the fully-adaptive source-oblivious model. 
\end{theorem}

Note that this result contrasts with the current knowledge about weaker variants of the source-oblivious model, like those defined in~\cite{DiksP96}. In particular, in the \emph{adaptive} variant (where nodes are not aware whether their neighbors received the broadcast message, apart from the neighbors from which they actually received the message), and in the \emph{non-adaptive} variant (where the nodes forward the message blindly, by following the orders specified by their lists, and ignoring the fact that there is no need to send the message to neighbors from which they actually received that message), the best known upper bound on the minimum number of rounds required to broadcast in the complete graph $K_n$ is $\log_2n + O(\log\log n)$~\cite{KimC05}. 

\medskip

Next, we focus on the construction of minimum broadcast graphs, and establish the following. 

\begin{theorem}\label{theo:MBG}
In the fully-adaptive source-oblivious model, for every $n\geq 1$, there are $n$-node broadcast graphs with $O(n \cdot L(n))$ edges. 
\end{theorem}

It follows from this result combined with from Eq.~\eqref{eq:B} that fully-adaptive source-oblivious broadcast protocols enable the design of broadcast graphs as sparse as what can be achieved with general broadcast protocols (i.e., protocols taking into account the source of the broadcast), while drastically reducing the space complexity of the broadcast table to be stored at each node, from $O(n\log d!)$ to $O(d\log d)$ bits at degree-$d$ nodes in $n$-node graphs. 

\paragraph{Remark.}

The graphs whose existence is guaranteed by Theorem~\ref{theo:MBG} are broadcast graphs, and therefore Theorem~\ref{theo:BG} can be viewed as a corollary of Theorem~\ref{theo:MBG}. However, the graphs exhibited in the proof of Theorem~\ref{theo:BG} have maximum degree $O(\log n)$, whereas the graphs used in the proof of Theorem~\ref{theo:MBG} have maximum degree as large as $\Omega(n)$, which may be an issue from a practical perspective, as far as network design is concerned. On the other hand, the graphs in the proof of Theorem~\ref{theo:BG} have $O(n\log n)$ edges, while the graphs in the proof of Theorem~\ref{theo:MBG} are sparser, and actually have the smallest possible number of edges. In particular, for every $k\geq 0$, and every $n$ satisfying $2^k< n\leq 2^k+2^{k-1}$, the $n$-node broadcast graphs in Theorem~\ref{theo:MBG}  have a linear number of edges, as $L(n)=1$ for $n$ in this range.  

\subsection{Related Work}
\label{subsec:relatedwork}

For more about the broadcast problem in general, we refer to the many surveys on the topic, such as~\cite{FraigniaudL94,HedetniemiHL88,Hromkovic05}. We provide below a quick survey of the literature dealing with broadcast under \emph{universal lists}, i.e., source-oblivious broadcast.

The broadcast problem under universal lists was first discussed indirectly by Slater et al. \cite{SlaterCH81}. 
The first formal definition of the problem of broadcasting with universal lists was given by Rosenthal and Scheuermann~\cite{RosenthalS87}, who described an algorithm for constructing optimal broadcast schemes for trees under the adaptive model. Later, Diks and Pelc~\cite{DiksP96} distinguished between non-adaptive and adaptive models with universal lists, and formally defined them. They designed optimal broadcast schemes for paths, cycles, and grids under both models. They also gave tight upper bounds for tori and  complete graphs, for adaptive and non-adaptive models. Diks and Pelc also described an infinite family of graphs for which the adaptive broadcast time is strictly larger than the broadcast time. Later, Kim and Chwa~\cite{KimC05} designed non-adaptive broadcast schemes for paths and grids. They also came up with upper bounds for hypercubes, and improved the upper bound from \cite{DiksP96} under non-adaptive model. More recently, the lower and upper bounds for trees under the non-adaptive model were tightened by Harutyunyan et al.~\cite{HLMS11}, as well as the upper bounds on general graphs. Harutyunyan et al. also presented a polynomial-time dynamic programming algorithm for finding the non-adaptive broadcast time of any tree. The most recent papers on the matter, by Gholami and Harutyunyan~\cite{GholamiH22a,GholamiH22b} defined the fully adaptive model with universal lists. Under this model they computed the broadcast time of grids, tori, hypercubes, and Cube Connected Cycles (CCC).

\section{Preliminaries} 

Recall that, according to the definition of fully-adaptive source-oblivious broadcast introduced in \cite{GholamiH22a}, once a vertex $v$ is informed, say at round~$t$, it will follow its list of neighbors $\ell_v$ (cf. Eq.~\eqref{eq:thelists}), and pass the message to the first vertex on the list which is not already informed before round $t+1$. In other words, not only node~$v$ skips all its neighbors from which it  received the message, but $v$ also skip all other informed neighbors. Given a set $L=(\ell_v)_{v\in V(G)}$ of lists, the number of rounds for broadcasting a message from $s\in V(G)$ to all the other nodes of~$G$ using lists~$L$ under the fully adaptive model is denoted by $b_{\fa}(G,s,L)$. 
The broadcast time of a graph~$G$ under the fully adaptive model is then defined as
\[
b_{\fa}(G)=\min_L\max_{s\in V(G)}b_{\fa}(G,s,L).
\]
%
%
%
The following results illustrate the above definition, and may serve as a warm up for the rest of the paper.  The first proposition shows that the broadcast time using protocols encoded with lists cannot be more than twice the broadcast time using general protocols. 

\begin{proposition}\label{prop:2b}
For every  graph $G=(V,E)$, the broadcast complexity of $G$ under the fully-adaptive source-oblivious model is at most $2\,\min_{s\in V}b(G,s)$, that~is,
\[b_{\fa}(G)\leq 2\; \min_{s\in V}b(G,s).\] 
\end{proposition}

\begin{proof}
Let $s\in V$ with minimum broadcast time among all nodes in~$G$, i.e., $b(G,s)=\min_{s'\in V}b(G,s')$. Let $T$ be a broadcast tree rooted at $s$, enabling broadcast from $s$ in $b(G,s)$ rounds. (Such a tree is a spanning tree of~$G$ rooted at~$s$ in which the children of any node~$v$ are ordered, specifying the order in which these children must be informed by~$v$ upon reception of a message from~$s$.) Let $v\neq s$ be a node of~$T$, let $w$ be the parent of $v$ in~$T$, and let $u_1,\dots,u_d$ be the $d$ children of $v$ in~$T$ enumerated in the order in which they are called in an optimal broadcast protocol from $s$ in~$T$. The list assigned to~$v$ is 
\[
\ell_v=(w,u_1,\dots,u_d).
\]
Similarly, the list assigned to~$s$ is 
\[
\ell_s=(u_1,\dots,u_d)
\]
where $u_1,\dots,u_d$ are the $d$ children of $s$ in~$T$ enumerated in the order in which they are called in an optimal broadcast protocol from $s$ in~$T$. 

To show that this set of lists enable fully adaptive broadcast to perform in at most $2b(G,s)$ rounds, let $u$ be a source node. Following the assigned universal lists every vertex makes the first call to its parent once gets informed. Thus, a message broadcast from~$u$ reaches the root~$s$ of~$T$ after at most $\mbox{depth}(T)$ rounds. Since $b(G,s)\geq \mbox{depth}(T)$, the message reaches~$s$ after at most $b(G,s)$ rounds. 
Once at~$s$, the message is broadcast down the tree in at most $b(G,s)$ rounds since, for every node~$v\neq s$, the parent~$w$ of~$v$ in the list $\ell_v$ is skipped in the fully adaptive model. The upward and downward phases amount for a total of $2\, b(G,s)$ rounds. 
\end{proof}

The following is a direct consequence of Proposition~\ref{prop:2b}.

\begin{corollary}
For every  graph $G$, the broadcast complexity of $G$ under the fully-adaptive source-oblivious model is at most $2\,b(G)$, i.e., $$b_{\fa}(G)\leq 2\,b(G).$$
\end{corollary}

Note that Proposition~\ref{prop:2b} also holds in the adaptive source-oblivious model (in this model every node~$v$ just skips the neighbors in~$\ell_v$ from which it received the message). The second proposition shows that there is a price to pay for using broadcast protocols encoded with lists, in the sense that the broadcast time using such protocols may be larger than the broadcast time using general protocols. 

\begin{proposition}\label{prop:lafamille}
There is an infinite family of graphs  $\mathcal{F}$ such that, for every $G\in \mathcal{F}$, the broadcast complexity of $G$ in the fully-adaptive source-oblivious model is larger than its broadcast time, that is $b_{\fa}(G)>b(G)$. 
\end{proposition}

\begin{figure}[tb]
  \centering
  \includegraphics[width=.5\linewidth]{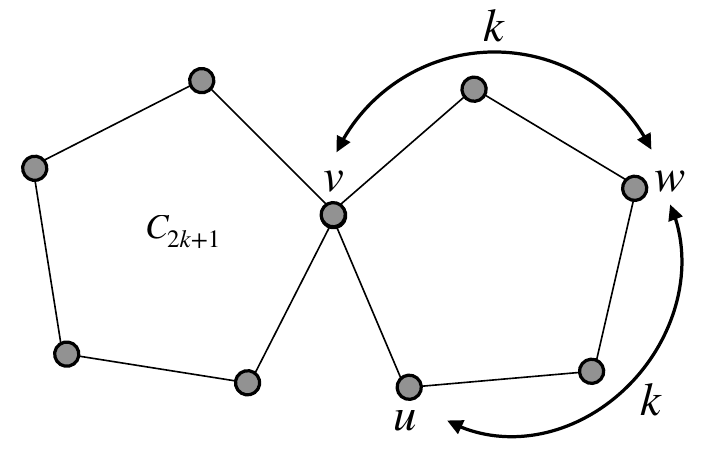}
  \caption{Graph in the proof of Proposition~\ref{prop:lafamille}, for $k=2$.}
  \label{fig:2cycless}
\end{figure}

\begin{proof}
    A basic example of such a family of graphs is obtained from two cycles $C_{2k+1}$ of length $2k+1$, by merging two of their vertices into one (see Fig.~\ref{fig:2cycless}). For every $k\geq 1$, the resulting graph~$G_k$ has $4k+1$ vertices, with a cut vertex~$v$ at the intersection of the two cycles. It is known~\cite{HedetniemiHL88} that $b(C_{2k+1})=k+1$, from which it follows that $b(G_k)=2k+1$. Now, in the fully-adaptive source-oblivious model, let $\ell_v$ be the list of node~$v$, and let~$u$ be the neighbor of~$v$ occurring first in this list. Let $w$ be the vertex antipodal to the edge $\{u,v\}$ in the cycle containing both $u$ and~$v$. If $w$ does not call first its neighbor on the shortest path from $w$ to~$v$ in~$G_k$, then broadcast from~$w$ will take at least~$2k+2$ rounds. On the other hand, if $w$ calls first its neighbor on the shortest path from $w$ to~$v$ in~$G_k$, then $v$ will receive the information in round~$k$. Since~$u$ has not yet received the information at the end of round~$k$, $v$~will proceed according to its list~$\ell_v$, and call~$u$ at round~$k+1$. As a result, $v$ will start broadcasting in the other cycle (the one not containing $u$ and $w$), no sooner than round $k+2$, and thus the whole protocol will not complete before $2k+2$ rounds. 
\end{proof}

\section{Broadcast Graphs} 

This section is entirely devoted to the proof of Theorem~\ref{theo:BG}. Recall that this theorem states that, in the fully-adaptive source-oblivious model,  there are $n$-node broadcast graphs for every $n\geq 1$, i.e., $n$-node graphs for which there exists a collection of lists $(\ell_v)_{v\in V(G)}$ achieving broadcast  in $\lceil\log_2n\rceil$ rounds from any source node.

\begin{proof}[Proof of Theorem~\ref{theo:BG}]
Recall that, for any $d\geq 1$, $b_{\fa}(Q_d)=d$ where $Q_d$ denoted the $d$-dimensional hypercube (see~\cite{GholamiH22b}) 
Therefore, it is sufficient to focus on graphs whose number of vertices is not a power of~2. So, let $n\ne 2^m$, with $\lceil \log_2 n \rceil = m$. Such an $n$ can be written as 
$$n = 2^{m-d_1} + 2^{m-d_2} + \dots + 2^{m-d_k},$$
where $1=d_1< d_2 < \dots < d_k \leq m$ for some $k\geq 2$. For such an~$n$, we consider the graph derived from the disjoint union of $k$ hypercubes 
$$Q_{m-1}, Q_{m-d_2},\dots,Q_{m-d_k}$$ 
with respective dimensions $m-1, m-d_2, \dots, m-d_k$. Let the vertices of $Q_{m-1}$ be labeled by the bit-strings 
$$0\alpha_1\dots\alpha_{m-1}$$ 
for all $\alpha_1\dots\alpha_{m-1} \in \{0,1\}^{m-1}$. Similarly, for every $i\in \{2, \dots,k\}$, let each of the $2^{m-d_i}$ vertices of hypercube $Q_{m-d_i}$ be labeled with an $m$-bit label having $d_i-1$ leading 1's followed by~0, and then one of the $2^{m-d_i}$ binary strings of length $m-d_i$ (see Fig.~\ref{fig:hypercubeBG}). Let us then define the graph $G=(V,E)$ where 
$$V=V(Q_{m-1})\cup V(Q_{m-d_2})\cup\dots\cup V(Q_{m-d_k})$$
and 
$$E=\{\{x,y\}\mid \mbox{$x$ and $y$ differ in exactly one bit}\}.$$ 
According to the labeling of the vertices defined above, every vertex in $V$ not in $Q_{m-1}$ is connected to a vertex of $Q_{m-1}$ by an edge in dimension~1 (i.e., the two vertices differs only in the first bit of their labels). More generaly, every vertex in $Q_{m-d_{i+1}}$ is connected to one vertex in each of the hypercubes $Q_{m-1}, Q_{m-d_2},\dots,Q_{m-d_i}$, by respective dimension $1, d_2,\dots,d_i$. Overall, $G$ has at most $\frac12n\lceil \log_2 n\rceil$ edges, as $G$ is a subgraph of the $m$-dimensional hypercube~$Q_m$. 

\begin{figure}[tb]
  \centering
  \includegraphics[width=.8\linewidth]{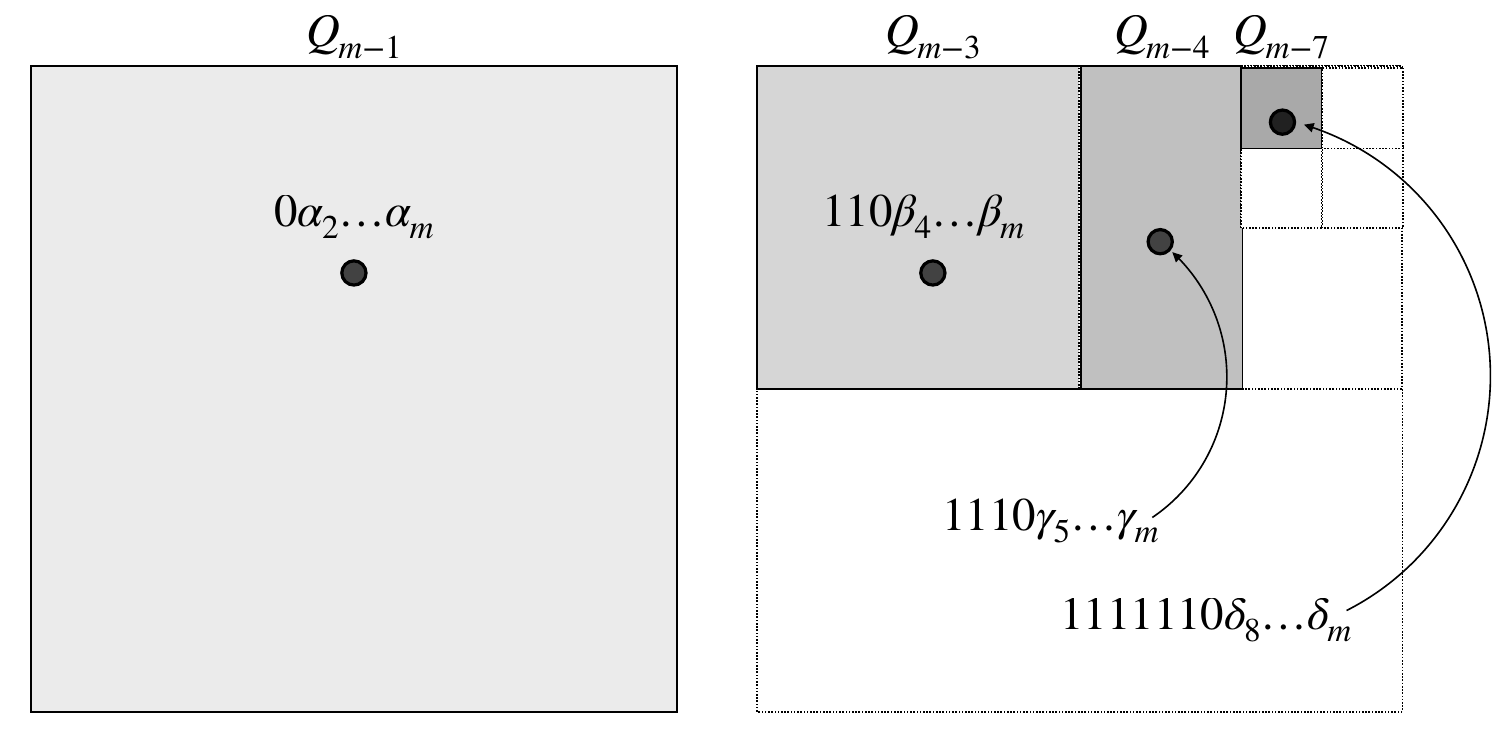}
  \caption{Construction in the proof of Theorem~\ref{theo:BG}. For every~$i$, $\alpha_i,\beta_i,\gamma_i,\delta_i\in\{0,1\}$. }
  \label{fig:hypercubeBG}
\end{figure}

\medskip

We assign the lists to the vertices as follows. Thanks to the structure of the graph~$G$, instead of using lists of neighboring vertices, we merely use list of integers, each one representing a \emph{dimension}, where, as in an $m$-dimensional hypercube, the neighbor of vertex $x_1\dots x_m$ in dimension $i$ is the vertex 
$$x_1\dots x_{i-1}\overline{x_i}x_{i+1}\dots x_m$$
with $\bar{x}=1-x$. For any vertex $u$ of $Q_{m-d_i}$ for some $1 \leq i \leq k$, its list is 
$$\ell_u = d_i+1,\dots,m,1,\dots,d_i.$$ 
So, in particular, for every vertex $u$ of $Q_{m-1}$, its list is $\ell_u = 2,3,\dots,m,1$. However, these lists must be used with the additional rule that if a node has no neighbors in the specified dimension, then this dimension is skipped, and the next dimension is considered. For example, a node $0\alpha_2\dots\alpha_m$ of $Q_{m-1}$ may not have a neighbor in dimension~1. This is for instance the case of every node $00\alpha_3\dots\alpha_m$ in the graph of Fig.~\ref{fig:hypercubeBG} because there is no hypercube $Q_{m-2}$ in this graph, and thus no nodes labeled $10\alpha_3\dots\alpha_m$. On the other hand, a node $010\alpha_4\dots\alpha_m$ in this graph has a neighbor in dimension~1, namely node $110\alpha_4\dots\alpha_m$, which appears in~$Q_{m-3}$. 

\medskip

We are now going to prove that these lists enable to establish that $b_{\fa}(G) = \lceil \log_2 n\rceil$. 
If the source $u$ is in $Q_{m-1}$ then, by following the  lists $\ell_v$ for all $v\in Q_{m-1}$, all vertices of $Q_{m-1}$ will be informed by round $m-1$. This can be checked easily, but this is a mere consequence of the fact that any hypercube is a broadcast graph under the fully-adaptive model~\cite{GholamiH22b}. At  round~$m$,  all the vertices of $Q_{m-1}$ that are connected to another vertex outside of $Q_{m-1}$ sends the information along dimension~1. Thus, at round~$m$, every vertex~$v$ of hypercube $Q_{m-d_i}$ with $m$-bit label $11\dots 10\alpha_1\alpha_2\dots \alpha_{m-d_i}$ will receive the message from the vertex of $Q_{m-1}$ with label $01\dots 10\alpha_1\alpha_2\dots\alpha_{m-d_i}$. 

If the source $u$ is in $Q_{m-d_j}$ for some $j>1$, then  its label is of the form $1\dots 10\alpha_1\dots \alpha_{m-d_j}$. As for $u\in Q_{m-1}$, according to  their lists 
$$\ell_v = d_j+1,\dots,m, 1,\dots,d_j,$$ 
all vertices~$v$ of $Q_{m-d_j}$ have received the information after $m-d_j$ rounds. Starting from round $m-d_j+1$ until  round~$m$, all informed vertices of $Q_{m-d_j}$ will then inform their neighbors in all the other hypercubes. 
\begin{itemize}
\item For the hypercubes $Q_{m-1},\dots,Q_{m-d_{j-1}}$, at round $m-d_j+1$, all $2^{m-d_j}$ vertices of $Q_{m-1}$ with labels of the form $01\dots 10\alpha_1\alpha_2\dots\alpha_{m-d_j}$ receive the message from their respective neighbor $11\dots 10\alpha_1\alpha_2\dots\alpha_{m-d_j}$  in $Q_{m-d_j}$. During the next $d_j-1$ rounds, these vertices will complete broadcasts in parallel in disjoint $(d_j-1)$-dimensional sub-cubes of $Q_{m-1}$ following the order $2,3,\dots,d_j$ of their lists. Similarly, at round $m-d_j+2$, all $2^{m-d_j}$ vertices of $Q_{m-d_2}$ with labels of the form $1\dots 101\dots 10\alpha_1\alpha_2\dots\alpha_{m-d_j}$ receive the message from their respective neighbor $1\dots 111\dots 10\alpha_1\alpha_2\dots\alpha_{m-d_j}$ in $Q_{m-d_j}$, via dimension~$d_2$. During the next $d_j-d_2\leq d_j-2$ rounds, these vertices will complete broadcasts in parallel in disjoint $(d_j-d_2)$-dimensional sub-cubes of $Q_{m-d_2}$ following the order $d_2+1,\dots,d_j$ of their lists. The same holds for all hypercubes $Q_{m-1},\dots,Q_{m-d_{j-1}}$. 

\medskip

\item The case of the hypercubes $Q_{m-d_{j+1}},\dots,Q_{m-d_k}$ of lower dimension is easier, for it is similar to the case where the source node is in~$Q_{m-1}$. At round $m-d_j+j\leq m$ all vertices $1\dots,10\alpha_1\alpha_2\dots\alpha_{m-d_j}$ of $Q_{m-d_j}$ will inform the at most $2^{m-d_j}$ nodes $Q_{m-d_{j+1}},\dots,Q_{m-d_k}$ of the form $1\dots,11\alpha_1\alpha_2\dots\alpha_{m-d_j}$, via dimension~$d_j$, completing broadcast in the lower dimensional hypercubes. 
\end{itemize}
It follows that $b_{\fa}(G) = \lceil \log_2 n\rceil$, as claimed.
Finally, the fact that, for every ${n\geq 1}$, the broadcast time of the clique~$K_n$ is $b_{\fa}(K_n)=\lceil\log_2n\rceil$ in the fully-adaptive source-oblivious model directly follows from the fact that the graph~$G$ used to establish the theorem is a subgraph of~$K_n$. 
\end{proof}

\section{Minimum Broadcast Graphs} 

This section is entirely devoted to the proof of Theorem~\ref{theo:MBG}. Recall that this theorem states that, in the fully-adaptive source-oblivious model, for every $n\geq 1$, there are $n$-node broadcast graphs with $O(n \cdot L(n))$ edges.

\begin{proof}[Proof of Theorem~\ref{theo:MBG}]
Recall that the $d$-dimensional hypercubes 
are minimum broadcast graphs for the fully adaptive model, i.e.,  for $n=2^d$, $B_{\fa}(n) = \frac12n \log n$ (see~\cite{GholamiH22b}).
For any $n$ not a power of~2, we will construct a graph on $n$ vertices and $O(nL(n))$ edges which is a broadcast graph under the 
fully adaptive model. More precisely, our graph has $n(L(n)+1)$ edges. Its construction is directly inspired from the construction in~\cite{GrigniP91}. Any $n$ not a power of~2 can be presented as 
$$n = 2^m - 2^k - r,$$ 
where $0 \leq k \leq m-2$ and $0 \leq r \leq 2^k - 1$. 

We begin the construction of our graph $G = (V,E)$ using $m-k$ binomial trees $T_{m-1},T_{m-2},\dots,T_{k}$ of sizes 
$2^{m-1},2^{m-2},\dots,2^k$, respectively rooted at vertices $v_{m-1},v_{m-2},\dots,v_{k}$. Recall that the binomial tree of size~1 consists of a single node (which is the root of the tree), and, for $d>0$, the binomial tree of size $2^d$ is obtained by connecting the two roots of two copies of a binomial tree of size $2^{d-1}$ by an edge, and selecting one of these two roots as the root of the resulting tree. The union of the trees $T_{m-1},T_{m-2},\dots,T_{k}$ contains $2^m - 2^k$ vertices and $2^m - 2^k - (m-k)$ edges. 

Next, we delete $r$ vertices (and $r$ edges) from $T_{k}$ by repeating $r$~times the removal of a leaf that is furthest away from the root of~$T_{k}$. In order to simplify the notation, we will abuse notation, and still call the resulting tree~$T_{k}$.  Note that since $r \leq 2^k - 1$, $T_k$ is not empty. 
This union of the trees $T_{m-1},\dots,T_{k}$ now contains $n= 2^m - 2^k - r$ vertices, and $2^m - 2^k - r - (m - k) = n-(m-k)$ edges. To complete the construction of our graph~$G$, we connect every vertex of $V=\cup_{i=k}^{m-1}V(T_i)$ to all the roots $v_{m-1},\dots,v_{k}$ of the $m-k$ binomial trees. Thus, the graph $G=(V,E)$ will have 
$$|V| = n = 2^m - 2^k - r$$ 
vertices, and 
$$|E| = n-(m-k) + (m-k)(n-1) = (m-k+1)n - 2(m-k)$$ 
edges.

To show that $b_{\fa}(s) = \lceil \log_2 n \rceil$ for any originator $s \in V$, we first assign the lists of all vertices in~$V$. Observe that each vertex $u$ of $G$ is actually the root of some binomial tree~$T_{m-p}$ for some $1 \leq p \leq m$ --- by definition of binomial trees. Let us denote $u=\root(T_{m-p})$. ALso, for the root~$w$ of a binomial tree of dimension at least $m-(p+1)$, let $c_{m-p}(w)$ be the child of~$w$ that is the root of a binomial tree of dimension~$m-p$. 
The list $\ell_{u}$ assigned to a vertex $u = c_{m-p}(w)$ distinct from $v_{m-1},\dots,v_k$ is 
\begin{equation}\label{eq:list-of-u}
 \ell_{u} = v_{m-1}, v_{m-2}, \dots, v_{m-p}, c_{m-p-1}(u), c_{m-p-2}(u), \dots , c_{0}(u).
\end{equation}
For a root vertex $u = v_{m-i}$ for some $1 \leq i \leq m-k$, we set 
\begin{equation}\label{eq:list-of-root-v}
\ell_{v_{m-i}} = v_{m-1}, v_{m-2}, \dots, v_{m-i},  
c_{m-i-1}(v_{m-i}), 
c_{m-i-2}(v_{m-i}), \dots, c_{0}(v_{m-i}).
\end{equation}
These lists have a desirable elementary property: 

\begin{fact}\label{basic-remark}
Any vertex $w$ that is the root of binomial tree of dimension~$d$ will inform all the vertices of its binomial tree during the last $d$ rounds of broadcast. Indeed, each of $w$'s children $c_{l-1}(w), c_{l-2}(w), \dots, c_{0}(w)$ will receive the message from $w$ at round $m-l+1, m-l+2, \dots, m$, respectively, and, by following their own lists, each will complete broadcast within its binomial trees of respective dimensions $d-1, d-2, \dots, 0$. As a result, all the vertices in the binomial tree of dimension $d$ rooted at $w$ will be informed by round $m$. 
\end{fact}

\begin{figure}[tb]
  \centering
  \includegraphics[width=.9\linewidth]{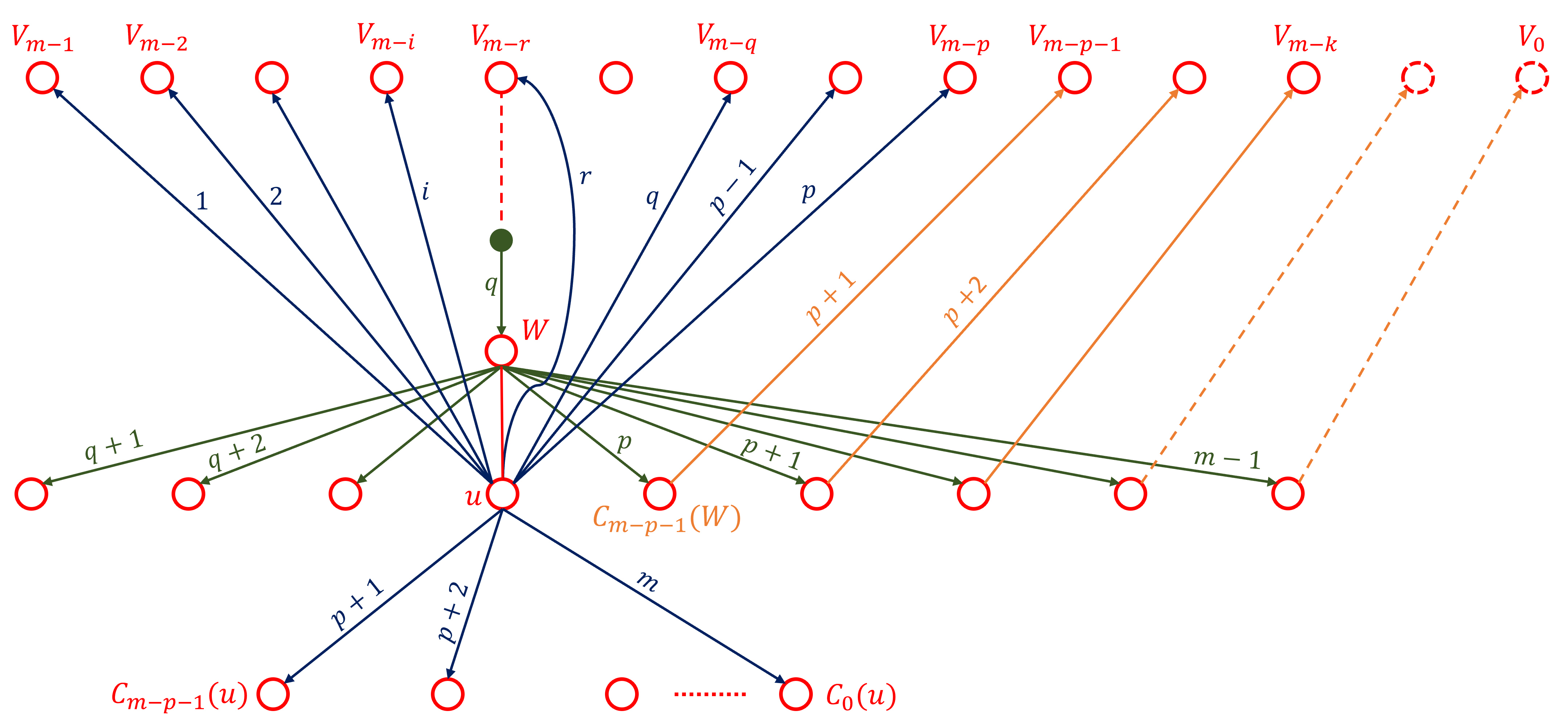} 
  \caption{Broadcast from source node $u$}
  \label{fig:treeMBG1}
\end{figure}

More generally, let us consider a vertex $u\in V$ as the source of broadcast, and let us assume that $u$ is the root of a binomial tree of dimension~$m-p$ (see Figure ~\ref{fig:treeMBG1}). Also, as $u$ belongs to one of the trees $T_{m-1},\dots,T_{k}$, let us assume that it belongs to tree $T_{m-r}$ rooted at vertex $v_{m-r}$, where $m-r > m-p$. Following its  list defined in Eq.~\eqref{eq:list-of-u},  vertex~$u$ first informs the root vertices $v_{m-1},\dots, v_{m-p}$ during the first $p$ rounds, and then it informs all of its $m-p$ children in its binomial tree of dimension $m-p$ during rounds $p+1, p+2,\dots, p+(m-p)$, completing after $m$ rounds.  

All root vertices $v_{m-1}, v_{m-2},\dots, v_{m-p}$ will receive the information from $u$ at rounds $1,2,\dots,p$, respectively, and will act according to their lists defined in Eq.~\eqref{eq:list-of-root-v}. In particular, following its list $\ell_{v_{m-1}}$, node~$v_{m-1}$ will skip $v_{m-1}$, and will inform all of its children within the binomial tree $T_{m-1}$ during the remaining $m-1$ rounds. Thus, all vertices of tree $T_{m-1}$ will be informed by round~$m$ thanks to Fact~\ref{basic-remark}. In general, any root vertex $v_{m-q}$ for $q=1,2,\dots,p$, except the root vertex $v_{m-r}$,  will receive the message from $u$ at round~$q$, and will follow its list $\ell_{v_{m-q}}$ specified in Eq.~\eqref{eq:list-of-root-v}. Since by round~$q$ all root vertices $v_{m-1}, v_{m-2},\dots,v_{m-q-1}$ are already informed by $u$, the root $v_{m-q}$ will skip the vertices $v_{m-1},\dots, v_{m-q}$ in its list, and, starting from round $q$, it will inform all its $m-q$ children $c_{m-q-1}(v_{m-q}), c_{m-q-2}(v_{m-q}), \dots, c_{0}(v_{m-q})$. Again, all vertices of the binomial tree $T_{m-q}$ will be informed by time unit $m$ thanks to Fact~\ref{basic-remark}.

Next, let us describe how broadcast proceeds in the binomial tree~$T_{m-r}$, as well as in the binomial trees $T_{m-p-1}, T_{m-p-2}, \dots, T_{k}$. Since $m-r > m-p$, the root $v_{m-r}$ of tree $T_{m-r}$  will receive the information from $u$ at round $r$, and, as mentioned above, will inform all of its children in its binomial tree, starting at time unit $r+1$. Let us assume that the parent~$w$ of $u$ in $T_{m-r}$ receives the information at round~$s$, which means that $w$ is the root of a binomial tree of dimension~$m-s$. We have $m-p < m-s \leq m-r$. Vertex $w$ will follow its list, and  will inform all its children, starting from round $m-s+1$. Following its list, node~$w$ had to inform vertex~$u$ at round~$p$.  However, since $u$ is already informed, $w$~will skip~$u$, and will inform its children $c_{m-p-1}(w), c_{m-p-2}(w), \dots, c_{0}(w)$ one round earlier, at rounds $p, p+1, \dots, m-1$, respectively. 
Now, following its list $\ell_{c_{m-p-1}(w)} = v_{m-1}, v_{m-2}, \dots, v_{m-p}, v_{m-p-1}, c_{m-p-2}(u), \dots, c_{0}(u)$, vertex~$c_{m-p-1}(w)$ will skip all informed vertices $v_{m-1},v_{m-2}, \dots, v_{m-p}$ (they were all informed from $u$), and will send the information to the root vertex~$v_{m-p-1}$ at round~$p+1$. Starting from round~$p+2$, vertex $w$ will inform all of its children in its binomial tree of dimension~$m-p-1$, and will complete broadcast by round~$m$. 

\begin{figure}[tb]
  \centering
  \includegraphics[width=.9\linewidth]{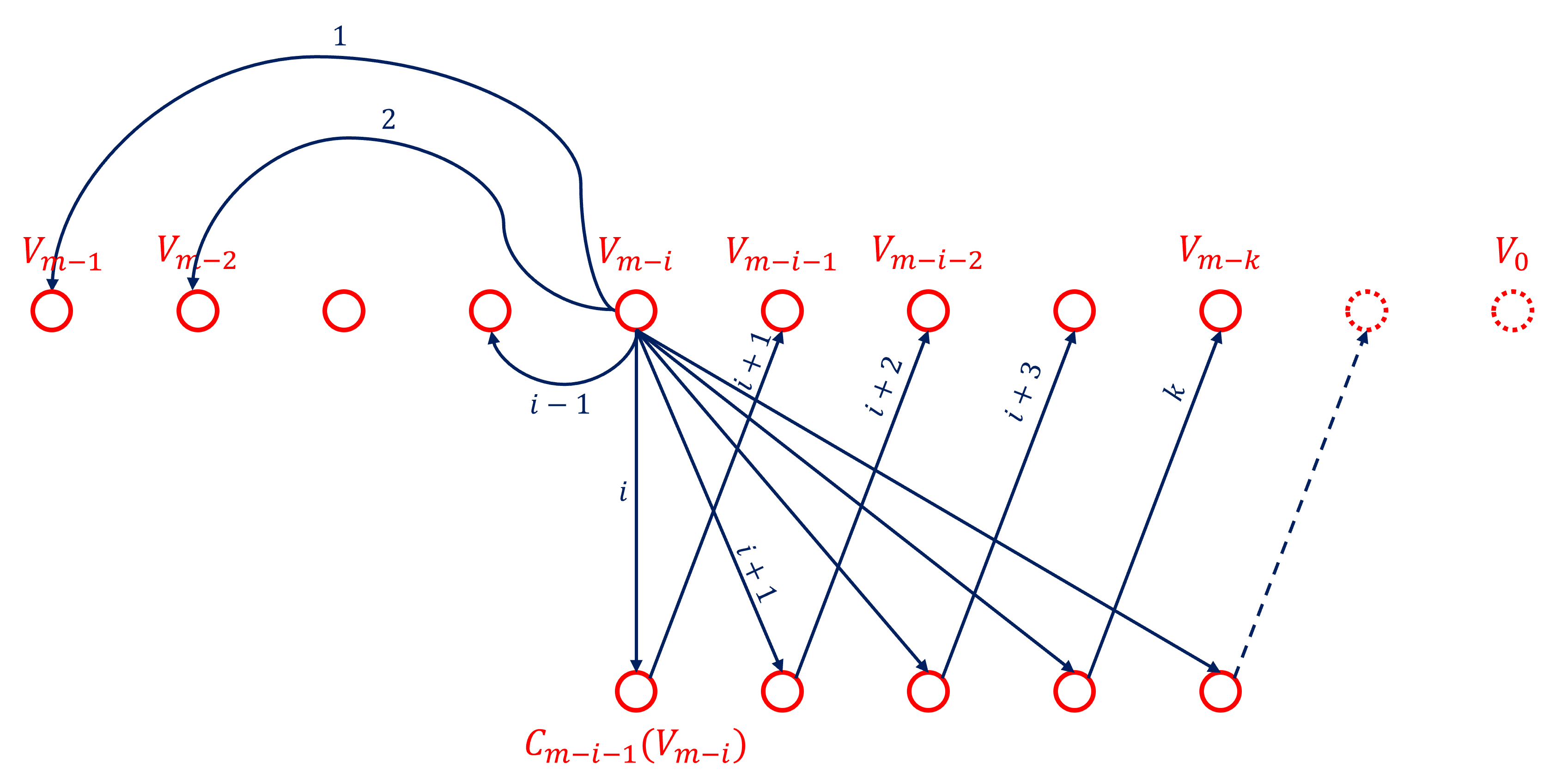} 
  \caption{Broadcast from source node $v_{m-i}$}
  \label{fig:treeMBG2}
\end{figure}

Similarly, each of $w$'s children, i.e., nodes $c_{m-p-t}(w)$ for $t=1,2,\dots,m-p-k$, will skip all informed vertices $v_{m-1}, v_{m-2}, \dots, v_{m-p}, \dots, v_{m-p-t}$, since they are informed earlier, either by $u$, or by the children $c_{m-p-1}(w), \dots, c_{m-p-t}$ of~$w$. Thus, $w$'s child $c_{m-p-t}(w)$ will send the information to  the root $v_{m-p-t}$ at round~$p+t$. Starting from round~$p+t+1$, vertex $w$ will then inform all of its children in its binomial tree of dimension $m-p-1$, and will complete broadcast by round~$m$. Thanks to Fact~\ref{basic-remark}, the Binomial tree $T_{m-p-t}$ will be fully informed by round~$m$ once its root vertex $v_{m-p-t}$ receives the message at round $p+t$ from $c_{m-p-t}(w)$ (see Figure~\ref{fig:treeMBG1}).

It remains to prove that all root vertices $v_{m-p-1}, \dots, v_{k}$ will complete broadcast within their respective binomial trees. This directly follows from fact~\ref{basic-remark}, as each root $v_{m-l}$, $p+1 \leq l \leq m-k$, receives the information from $c_{m-l}(w)$ at round~$l$, and completes broadcast in its Binomial tree $T_{m-l}$ during  the remaining $m-l$ rounds, for all $l=p+1, \dots, m-k$. 

Finally, if the source of broadcast is one of the root vertices $v_{m-1}, \dots, v_{k}$, then the process is even simpler. The details of broadcast from a root vertex source is displayed in Figure~\ref{fig:treeMBG2}. 
\end{proof}

\section{Conclusion} 

We have shown that, as far as the design of minimum broadcast graphs is concerned, the power of broadcast protocols is not limited by bounding them to be encoded by a single ordered list of neighbors at each node, which has the profitable feature of drastically reducing the space-complexity of the local encoding of the protocols. 

Our results hold under the assumption that every node can signal its neighbors to let them know that it has received the broadcast information. The cost of signaling the neighbors is negligible compared to the cost of transmitting a potentially long message, but getting rid of this assumption may be desirable, by focusing, e.g., on adaptive or even non-adaptive source-oblivious protocols. 
The analysis of adaptive and non-adaptive protocols however appears to be quite challenging. In fact, it is not even clear whether $n$-node broadcast graphs exists for all~$n$ under these constraints. Indeed, as already mentioned, the best known upper bound on the broadcast time of cliques is $\log_2n + O(\log\log n)$~\cite{KimC05}. On the other hand, a systematic study of the minimum broadcast graphs with small number of nodes show that optimal broadcast protocols for these graphs (i.e., protocols performing in $\lceil\log_2n\rceil$ rounds) can be implemented by lists in the adaptive source-oblivious model. So, it may actually be the case that the aforementioned signaling assumption can be removed, while essentially preserving the good properties of the protocols. This is however not clear, and we state that issue as an open problem. 

\paragraph{Open problem.}

Is there an infinite family of $n$-node graphs $(G_n)_{n\geq 1}$ such that, for every $n\geq 1$, the broadcast time of $G_n$ in the adaptive (or even non-adaptive) source-oblivious model is $\lceil\log_2n\rceil$? And, independently from whether the answer to the previous question is positive or not, what is the minimum number of edges of $n$-node graphs with optimal broadcast time in the adaptive (or non-adaptive) source-oblivious model?

\bibliographystyle{plain}

\end{document}